\documentclass[11pt,letterpaper]{article}
\usepackage{fullpage} 
\usepackage[parfill]{parskip}

\usepackage{amssymb}

\usepackage{natbib}

\usepackage{nicefrac}       

\usepackage{pgfplots}
\usepackage{tikz}
\usepackage{amsfonts,amssymb}
\usepackage[caption=false]{subfig}
\usetikzlibrary{topaths,calc}

\usepackage{amsmath,amsthm}

\usepackage{xfrac}

\usepackage{comment}
\usepackage{bbm}
\usepackage{hhline}
\usepackage{mathtools}
\usepackage{ifthen}

\usepackage{graphicx,color}
\usepackage{fancybox}

\usepackage{array,float}
\usepackage{url}
\usepackage{mathrsfs}

\definecolor{cornellred}{rgb}{0.7, 0.11, 0.11}
\definecolor{maroon}{rgb}{0.52, 0, 0}
\definecolor{dgreen}{rgb}{0.0, 0.5, 0.0}
\definecolor{ballblue}{rgb}{0.13, 0.67, 0.8}
\definecolor{royalblue(web)}{rgb}{0.25, 0.41, 0.88}
\definecolor{bleudefrance}{rgb}{0.19, 0.55, 0.91}
\definecolor{royalazure}{rgb}{0.0, 0.22, 0.66}

\usepackage{enumitem}
\usepackage{multirow}
\usepackage{changepage}

\usepackage{accents}
\usepackage{todonotes}

\usetikzlibrary{arrows, decorations.markings}

\definecolor{bostonuniversityred}{rgb}{0.8, 0.0, 0.0}
\usepackage{hyperref}
\hypersetup{
	colorlinks = true,
	linkcolor=blue!70!black,
	citecolor=blue!70!black,
	linkbordercolor = {white}
}
\usepackage{cleveref}

\usepackage{thmtools}
\theoremstyle{plain}
\newtheorem{theorem}{Theorem}[section]
\newtheorem{lemma}[theorem]{Lemma}
\newtheorem{claim}[theorem]{Claim}

\newtheorem{proposition}[theorem]{Proposition}
\newtheorem{corollary}[theorem]{Corollary} 

\usepackage{thm-restate}

\theoremstyle{plain}
\newtheorem{definition}{Definition}[section] 
\newtheorem{example}[definition]{Example}
\newtheorem{remark}[definition]{Remark}

\theoremstyle{plain}

\newcommand{\xhdr}[1]{\vspace{3pt}\noindent{\bf {#1.}}}

\newcommand{\omt}[1]{}

\newcommand{\squishlist}{
   \begin{list}{{{\small{$\bullet$}}}}
    { \setlength{\itemsep}{3pt}      \setlength{\parsep}{1pt}
      \setlength{\topsep}{1pt}       \setlength{\partopsep}{0pt}
     \setlength{\leftmargin}{1em} \setlength{\labelwidth}{1em}
      \setlength{\labelsep}{0.5em} } }
\newcommand{\squishend}{  \end{list}  }

\newcommand{\squishlistTmp}{
   \begin{list}{{{\small{$\bullet$}}}}
    { \setlength{\itemsep}{2pt}      \setlength{\parsep}{0.5pt}
      \setlength{\topsep}{0.5pt}       \setlength{\partopsep}{0pt}
     \setlength{\leftmargin}{0.5em} \setlength{\labelwidth}{0.5em}
      \setlength{\labelsep}{0.5em} } }
\newcommand{\squishendTmp}{  \end{list}  }

\newcounter{relctr} 
\everydisplay\expandafter{\the\everydisplay\setcounter{relctr}{0}} 

\AtBeginDocument{} 

\newcommand{\argmax}{\mathop{\mathrm{arg\,max}}}

\newcommand{\blackwellOrder}{\preceq}

\newcommand{\supp}{\cc{supp}}

\newcommand{\messagespace}{\Sigma}

\newcommand{\messageprob}{\phi}

\newcommand{\signal}{\sigma}

\newcommand{\boxNum}{N}
\newcommand{\prior}{H}
\newcommand{\priorPDF}{h}
\newcommand{\priorMean}{\lambda}

\newcommand{\DP}{\cc{DP}}
\newcommand{\OPT}{\cc{OPT}}

\newcommand{\infoStrategy}{strategy}
\newcommand{\infoStrategies}{strategies}

\newcommand{\searcherU}{u}

\newcommand{\condition}{\mid}

\newcommand{\prob}[2][]{\text{Pr}\ifthenelse{\not\equal{}{#1}}{_{#1}}{}\!\left[{\def\givenn{\middle|}#2}\right]}
\newcommand{\expect}[2][]{\mathbb{E}\ifthenelse{\not\equal{}{#1}}{_{#1}}{}\!\left[{\def\givenn{\middle|}#2}\right]}
\newcommand{\indicator}[2][]{\mathbf{1}\ifthenelse{\not\equal{}{#1}}{_{#1}}{}\!\left\{{\def\givenn{\middle|}#2}\right\}}

\newcommand{\cc}[1]{\ensuremath{\mathsf{#1}}} 

\usepackage{mdwlist}
\makeatletter
\newenvironment{myproof}[1][\proofname]{\par
  \pushQED{\qed}%
  \normalfont \topsep2pt \relax
  \trivlist
  \item[\hskip\labelsep
        \itshape
    #1\@addpunct{.}]\ignorespaces
}{%
  \popQED\endtrivlist\@endpefalse
}
\makeatother

\begin{document}

\title{Intrinsic Robustness of Prophet Inequality to \\
Strategic Reward Signaling}
\author{Wei Tang\thanks{Chinese University of Hong Kong. Email: {\tt weitang@cuhk.edu.hk}.}
\and Haifeng Xu\thanks{University of Chicago. Email: {\tt haifengxu@uchicago.edu}.}
\and Ruimin Zhang\thanks{University of Chicago. Email: {\tt ruimin@uchicago.edu}.}
\and Derek Zhu\thanks{University of Chicago. Email: {\tt dzhu8@uchicago.edu}.}}

\maketitle

\begin{abstract}

Prophet inequality concerns a basic optimal stopping problem and states  that simple threshold stopping policies --- i.e., accepting the first reward larger than a certain threshold --- can achieve tight  $\sfrac{1}{2}$-approximation to the optimal prophet value.  Motivated by its economic applications, this paper studies the robustness of this approximation to natural \emph{strategic manipulations} in which each random reward is associated with a self-interested player who may selectively reveal his realized reward to the searcher in order to  maximize his probability of being selected. 

We say a threshold policy is $\alpha$(-strategically)-robust if it  (a)  achieves the  $\alpha$-approximation 
to the prophet value for strategic players; and (b) meanwhile remains a  $\sfrac{1}{2}$-approximation in the standard non-strategic setting.
Starting with a characterization of each player's optimal information revealing strategy, we demonstrate the intrinsic robustness of prophet inequalities to strategic reward signaling through the following results:
(1) for arbitrary reward distributions, there is a threshold policy that is $\frac{1-\sfrac{1}{e}}{2}$-robust, and this ratio is tight;
(2) for i.i.d.\ reward distributions,
there is a threshold policy that
is $\sfrac{1}{2}$-robust, which is  tight for the setting; and (3) for log-concave (but non-identical) reward distributions, the
 $\sfrac{1}{2}$-robustness can also be achieved under certain regularity assumptions.

\end{abstract}

\section{Introduction}

\newcommand{\prophet}{\cc{OPT}}

The  prophet inequality   of \citet{KUS-77} is a foundational framework for the theory of optimal stopping problems and sequential decision-making.
In the classic prophet inequality, a searcher faces a finite sequence of non-negative, real-valued and independent random variables $X_1, \ldots, X_\boxNum$ 
with known distributions $\prior_i$ from which a reward of value $X_i$ is drawn sequentially for $i=1, \ldots, \boxNum$.
Once a random reward is realized, the searcher decides whether to accept the realized reward and stop searching, or reject the reward and proceed to the next reward.
The searcher's objective is to maximize the value of the accepted reward.  The performance of the searcher's stopping policy is evaluated against a \emph{prophet value}
which equals to the ex-post maximum realized reward.
The classic and elegant result of \citet{SC-84} showed that a simple static threshold policy achieves at least half of the prophet value, and surprisingly, this bound is the best possible even among dynamic policies. Samuel-Cahn's policy uses the threshold that is the median of the distribution of the highest prizes, and then accepts the first realized reward that exceeds this threshold.
The existence of this $\sfrac{1}{2}$-approximation is now known as the {\em prophet inequality}. 

Recently, there is a regained interest of the prophet inequality due to its beautiful 
connection to online mechanism design (see, e.g., \cite{HKS-07,CHMS-10}), 
and many different settings of the prophet inequality have been studied (see the survey by \cite{A-17,C-19}). For example, \citet{KW-12} extend the prophet inequality to all matroid constraints and similar to \cite{SC-84}, they also show that a threshold stopping policy with the threshold equal to half of the expected maximum reward can also lead to the optimal $\sfrac{1}{2}$-approximation. Indeed, it is now well-known that there 
exists a range of thresholds for the classic prophet inequality that can achieve the optimal $\sfrac{1}{2}$-approximation (see \Cref{defn:1/2 threshold}).

An important assumption in the classic prophet inequality is that the distribution of each random variable is an inanimate object and, once the searcher reaches it, it will fully disclose its realized reward to the searcher. 
Yet this may not be the case in many real-world applications where each distribution 
may often be associated with a {\em strategic player}\footnote{We refer to the one that controls the reward information as the ``player'', and refer to the one that decides when to stop as the ``searcher''.}  
who may have incentives to selectively disclose information to maximize his own probability of being chosen by the searcher. 
This is usually the case when information is not controlled by nature but by humans or algorithms. Such examples are ample in economic activities. For instance,  when a recruiter searches for the best candidate for a position by sequentially interviewing a set of job applicants, each job applicant is naturally a strategic player and would want to control how much information they disclose about their characteristics (including strengths, weaknesses, personality, and experience) to the recruiter (the searcher) so as to maximize the probability of being hired. Similarly, when a venture capitalist searches for the most promising startup to invest by sequentially visiting each startup, the startups are strategic players. They 
have more accurate information about their own potential and can control how much of this information they disclose to attract the investment from the venture capitalist  (the searcher).
Finally, in the real estate market, when a buyer searches for the most attractive property to purchase,   the property sellers are strategic players who can control the amount of information they disclose about the property's condition and potential returns to the buyer  (the searcher).

Motivated by real-world applications like above, we introduce and study 
a natural variant of the prophet inequality   where each reward distribution is associated with a strategic player who can decide what information about the realized reward will be disclosed to the searcher. 
To capture such partial information revelation, we adopt the standard framework of information design  (also known as Bayesian persuasion \citep{KG-11,BM-16}) and assume that each player can selectively disclose reward information by implementing an information  strategy -- 
often referred to   as an \emph{experiment} or \emph{signaling scheme} \citep{KG-11,BM-16} --  which stochastically maps 
the realized reward (unobservable to the searcher) to a random signal (observable to the searcher). For the motivating applications above, such revelation strategies  encode  answers to standard interview questions (e.g., experiences, behavioral questions) prepared in advance by job candidates, the presentations prepared by startups for the VC, or the property's brochures prepared by property sellers. 
Each player aims to maximize his probability of being chosen by the searcher,
leading to a (constant-sum) competition among them. 
\xhdr{Our Results} In this work, we stand in the searcher's shoes and look to  understand how robust the classic prophet inequalities are to the above strategic player behaviors. Like most previous work in this space, we restrict our attention to threshold stopping policies.
 We start by characterizing players' optimal information revealing strategies. Given any threshold stopping policy with threshold $T$, the optimal information 
strategy of player $i$ with prior reward distribution $H_i$  has the following clean threshold structure (albeit using a threshold different from $T$): 
there exists a reward cutoff $t_i$ such that   player $i$ simply reveals whether $X_i \geq t_i$ or $X_i < t_i$. Moreover, $t_i$ satisfies $\expect{   X_i| X_i \geq t_i}  = T$. In words, player $i$ simply pools all the ``good rewards'' together to make their expectation barely pass the threshold $T$.\footnote{More formally, these response strategies form a subgame perfect Nash equilibrium among the sequential-move game of players (which turns out to be unique as far as players' utilities are concerned) (see \Cref{rmk:SPNE}).} This characterization allows us to reduce players' strategic behaviors to a related prophet inequality problem with binary reward supports. Our later analysis hence only needs to further investigate how well the thresholds for classic prophet inequalities perform in this (related) additional problem. 

Armed with the above characterization, next we turn to analyze the intrinsic robustness of the classic prophet inequality.  We say a   threshold stopping policy with threshold $T$ is $\alpha$(-strategically)-robust if it retains the $\sfrac{1}{2}$-approximation in the classic setting and meanwhile also can achieve
$\alpha$-approximation in the strategic scenario when all the players optimally reveal information. Our first main result is that, for arbitrary reward distributions, the threshold   policy with $T$ equaling the \emph{half expected max} threshold of \citet{KW-12} is 
$\frac{1-\sfrac{1}{e}}{2}$-robust  (\Cref{thm:general prior}). Moreover,  this competitive ratio of $\frac{1-\sfrac{1}{e}}{2} \approx  0.316 $  is the best among all known threshold policies that can secure the $\sfrac{1}{2}$-approximation in the classic non-strategic setting (\Cref{lem:tightness general priors}). 
This suggests that this well-studied threshold  policy
can also perform robustly well even when players are all strategic, illustrating the intrinsic robustness of the classic prophet inequality.

When the reward distributions are identical 
--  referred to as  \emph{IID distributions} which have been extensively studied in literature
\citep{PST-22,CFHOV-17,AKW-14,HHS-20,AEEHKL-17,HK-82} --  
we show that there exists a threshold in the range of known thresholds (\Cref{defn:1/2 threshold}) that is  $\sfrac{1}{2}$-robust  (\Cref{thm:IID prior}). Moreover, this $\sfrac{1}{2}$ approximation ratio is optimal among all possible threshold stopping policies for any IID distributions (\Cref{prop:IID tightness}).
Finally, when the reward distributions are not identical but log-concave, under certain regularity conditions  we show that any threshold between the  expected max and the median of the highest reward  \citep{SC-84} is $\sfrac{1}{2}$-robust (\Cref{thm:mainOblivious}). 
Note that log-concave reward distributions have been considered
in previous works on prophet inequality (see, e.g., \cite{AKW-14,CFHOV-17}); moreover,  it is satisfied by a wide range of distributions (e.g., normal, uniform, Gamma, Beta, Laplace, etc., \cite{BB-06}) and is 
a widely adopted assumption  in algorithmic game theory (see, e.g., \cite{CHK-07}).

\xhdr{Additional Related Work}
Prophet inequality is
a fundamental problem in optimal stopping theory which
was introduced in the 70s \cite{KS-78,KUS-77}.
Recently there has been a growing interest in prophet inequalities, generalizing the problem to different settings \citep{A-14,BMNK-07,CFHOV-17,DJSW-11,DSA-12,KW-12,CDL-23,AK-22,AKW-14,DMZ-22}. Our discussion here cannot do justice to its rich literature; hence, we refer interested readers 
to the recent survey by \citet{A-17,C-19} for   comprehensive overviews of 
recent developments and its connections to economic problems. 
This work  takes an informational perspective and associate 
each distribution in the prophet inequality with a strategic player 
that strategically signals reward information to the searcher. 
We follow the information design literature \citep{KG-11,BM-16}
and allow   players to  design  signaling schemes to influence 
the searcher's beliefs about the realized rewards. 
In our considered game, players are competing with each other 
for the selection of the searcher. The players' game thus 
also shares similarity with 
competitive information design \citep{GK-16,GK-17,AK-20,DHFTX-23,hossain2024multi}.

Conceptually, our work also relates to the rapidly growing literature on strategic machine learning (see, e.g., \cite{ZC-21,SMPW-16,GNE-21,DRS-18,SVXY-23,ZMS-21,BLW-21}), which studies learning  from strategic data providers. We also study similar strategic reward providers, albeit in a different algorithmic problem, i.e., optimal stopping.
More generally, our work  subscribes to the literature on   information design in sequential decision-making.
In particular, our paper relates to the recently increased interests 
on using online learning approaches to study the regret minimization 
when an information-advantaged player repeatedly interacts with an 
information-disadvantaged player \citep{BCC-23,CCMG-20,FTX-22,WZFWYJX-22,ZIX-21} without knowing their preferences. Instead of focusing on regret minimization, we use the competitive ratio (as conventionally used in prophet inequalities) to measure the searcher's policy.
It is worth mentioning that \cite{HHS-20} also studies strategic information revealing in prophet inequality problems. In their setting, a centralized player strategically discloses reward information to the searcher, while players in our setting form a decentralized game and each player acts on his own to maximize his payoff.

\section{Preliminary}

\label{sec:setting}

\newcommand{\equilStrat}{G^*}
\newcommand{\bestResponseG}{G^{\cc{BR}}}
\newcommand{\boxU}{u^B}
\newcommand{\HEMthreshold}{\theta^{\cc{HEM}}}

\newcommand{\policy}{q}
\newcommand{\stopIdx}{\mathbb{A}}

\newcommand{\medianThresholdUB}{T_{\cc{SC}}}
\newcommand{\medianThresholdLB}{\underline{T}_{\cc{SC}}}
\newcommand{\halfMax}{T_{\cc{KW}}}
\newcommand{\TEqual}{T^*}

\newcommand{\searcherUS}{\searcherU^{\cc{s}}}
\newcommand{\searcherUNS}{\searcherU^{\cc{ns}}}

In this section, we first revisit the formulation of the classic prophet inequality problem, and then
formally introduce our setting as its natural variant where each distribution is associated with a strategic player who will be strategically signaling their reward information.

\xhdr{Classic prophet inequality} In standard settings, a 
searcher faces a finite sequence of known
distributions $\prior_{1:\boxNum} \triangleq  
(\prior_i)_{i\in[\boxNum]} $
of $\boxNum$ non-negative independent random variables.
The outcomes (i.e., the rewards) $X_i\sim\prior_i$  \footnote{Here $\prior_i$ denotes the cumulative distribution function (CDF) of the reward $X_i$} for $i\in[\boxNum]$
are revealed sequentially to the searcher for $i=1,\ldots, \boxNum$.
Let $\lambda_i\triangleq \expect[\prior_i]{X_i}$  denote the mean reward for distribution $\prior_i$.
Upon seeing a reward, the searcher 
 decides whether to accept the observed reward and stop the process, or irrevocably move on to the next reward.
The searcher's goal is to maximize the expected accepted reward.
The searcher's expected payoff cannot be more than that of a prophet who knows in advance the realizations of the rewards, namely, $X_1, \ldots, X_\boxNum$.
We denote by 
$\prophet \triangleq 
\expect[\prior_1, \ldots, \prior_\boxNum]{\max_{i\in[\boxNum]}X_i}$ the prophet value. 
Given a stopping policy $\policy$, let $X^{(q)}$
be the accepted reward.
The stopping policy $\policy$ is said to be $\alpha$-approximate (of the prophet value) if 
the following holds for the searcher's expected payoff for any input distributions $\prior_{1:\boxNum}$:
\begin{align*}
    \expect{X^{(\policy)}} \ge \alpha \cdot \prophet~.
\end{align*}
The above statement is often referred to as the prophet inequality. For any stopping policy $\policy$, the largest constant $\alpha \in (0,1]$  that satisfies the above inequality is named the competitive ratio.
Classic results of \cite{KS-78,KUS-77} elegantly show that 
there   exist simple threshold stopping policies
that achieve the competitive ratio $\sfrac{1}{2}$, and moreover the 
ratio of $\sfrac{1}{2}$ is tight.\footnote{See \cite{A-17,C-19} for the hardness example for the $\sfrac{1}{2}$-approximation.}
\begin{definition}[Threshold Stopping Policy]
\label{defn:threshold}
A threshold stopping policy is an online algorithm which
pre-computes a threshold value $T$ as a function of the distributions $\prior_{1:\boxNum}$, and then accepts the first reward $X_i$
whose value is no smaller than the threshold, i.e., $X_i \ge T$.
\end{definition}
It is well-known that multiple thresholds can achieve the optimal $\sfrac{1}{2}$-approximation ratio.
\begin{definition}[$\sfrac{1}{2}$-approximation Threshold Spectrum]
\label{defn:1/2 threshold}
Let $\halfMax \triangleq \sfrac{\expect{\max_i X_i}}{2}$ (\citet{KW-12}), 
let $\TEqual$ satisfy 
$\TEqual = \sum_{i\in[\boxNum]}\expect{(X_i - \TEqual)^+}$, and
let 
$\medianThresholdUB \triangleq 
\sup\{t: \prob{\max_i X_i \ge t} \ge \sfrac{1}{2}\}$ (\citet{SC-84}).
Then any threshold stopping policy with threshold
$T\in [\halfMax,\max\{\medianThresholdUB, \TEqual\}]$
guarantees $\sfrac{1}{2}$-approximation to the prophet value.\footnote{
There also exists other thresholds that could lead to $\sfrac{1}{2}$-approximation, e.g., any $T\in[\min\{\medianThresholdLB,\halfMax\}, \halfMax)$ where $\medianThresholdLB \triangleq \inf\{t: \prob{\max_i X_i \ge t} \ge \sfrac{1}{2}\}$. However, using these thresholds requires to modify the policy defined in \Cref{defn:threshold} to be a {\em strict} stopping policy to obtain $\sfrac{1}{2}$-approximation (i.e., searcher accepts a first reward that is strictly larger than threshold).
Moreover, under strict stopping policy, there may exist no Nash equilibrium among players' game that we formulate shortly.
}
\end{definition}

    

\xhdr{Prophet inequality with strategic reward signaling}
In the strategic setting, each distribution $\prior_i$ may be associated with a strategic 
player that governs how much information about the realized reward he would 
like to reveal to the searcher once the searcher reaches his distribution.
Formally, 
upon arriving at the reward distribution $\prior_i$, 
the searcher does not directly observe the realized reward $X_i \sim \prior_i$. 
Instead, she observes an information signal, designed by the player $i$ with distribution $\prior_i$, that is correlated with the reward $X_i$.
We follow the literature in information design \citep{KG-11}
to model this strategic reward signaling:
Each player $i$  chooses a signaling scheme 
$\messageprob_i(\cdot \condition x)\in\Delta(\messagespace_i)$,
where $\messagespace_i$ is a  measurable signal space and 
$\messageprob_i(\signal \condition x) \in [0, 1]$
specifies the conditional probability of a signal $\signal \in \messagespace_i$ that will be sent to the searcher
when the reward $X_i = x \sim \prior_i$ is realized.
Notice that, upon seeing a signal $\signal \sim \messageprob_i(\cdot \condition x)$,
together with the prior information $\prior_i$ from which the reward is realized,
the searcher can update her Bayesian belief about the underlying realized value $X_i$,
and then decides whether to stop and choose   player $i$, or reject $i$ to continue  her search.

Each player $i$ is competing with each other for the final 
selection from the searcher. 
Namely, a player obtains payoff $1$ if his reward is accepted and payoff $0$ if his reward is not accepted.\footnote{All of our results hold if each player $i$ has payoff $v_i$ if his reward is accepted and payoff $u_i$ if his reward is not accepted as long as $v_i > u_i$.}
Each player's goal is to design an information revealing strategy that maximizes his probability of being chosen by the searcher. We below give two simple examples of information revealing strategies.
\begin{example}[Examples of Information Strategies]
(1) No information  strategy: 
the signal is completely uninformative 
(e.g., the distribution 
s$\messageprob_i(\cdot \condition x)$ is a Dirac delta function on a single signal, i.e., $|\messagespace_i| = 1$), hence the searcher infers an expected reward of $\lambda_i = \expect{X_i}$ as her perceived reward from player $i$;  
(2) Full information revealing strategy: 
the signal perfectly reveals player $i$'s value to
the searcher
(i.e., $\messageprob_i(\signal \equiv x\condition x) = 1$ 
for every realized $X_i = x$, and $\messagespace_i = \supp(\prior_i)$) 
\end{example}



In this work, standing in the searcher's shoes, we are interested in how threshold stopping policies perform under players' optimal strategic reward signaling. 

\xhdr{Game Timeline} 
 The timeline of our prophet inequality with strategic players problem, where the searcher employs a threshold stopping policy,
 can be detailed as follows: 
(1) Knowing $\prior_{1:N}$, the searcher first announces a threshold stopping policy with threshold $T$ that is a function of $\prior_{1:N}$;
(2) Knowing threshold $T$, each player then picks a signaling scheme (also known as an \emph{experiment} in economics literature \cite{KG-11,BM-16}) to reveal partial information about the underlying reward; and
(3) The searcher learns all players' information strategies, and
then conducts a search based on her threshold stopping policy with threshold $T$ (i.e., accepting the first player whose posterior mean of his reward distribution, given his revealed information signal, exceeds the threshold $T$). The assumption that the searcher knows information strategies as well as realized signals is commonly adopted in the information design literature \cite{KG-11,BM-16}, and is also well motivated for the domains of  our interest. For instance, when startups persuade VCs or property sellers persuade buyers, the signaling scheme could correspond to startups’ product exhibitions or sellers’ promotion brochures  which   determine what information the searcher could see. Misreporting realized signals corresponds to revealing untrue information, which not only violates regulation policies and but also causes the players to lose credibility in the long term.

Slightly abusing the notation, we also use $X^{(q)}$ to denote the accepted reward given the searcher's stopping policy $\policy$ under the strategic reward signaling, and let $\searcherUS(q)\triangleq \expect{X^{(q)}}$
be the searcher's expected payoff. Notice that here the expectation is 
not only over the randomness of the distributions $\prior_{1:\boxNum}$, 
but also the information strategies $\{\phi_i(\cdot\mid x)\}$.
Anticipating the players' strategic behavior, 
 the searcher wants a stopping policy that can still guarantee a good performance 
competing against the prophet value $\prophet$. 

 As mentioned earlier, we are particularly interested in 
how previously studied threshold stopping policies described  
in \Cref{defn:1/2 threshold} perform under the players' strategic reward signaling. To formalize our goal, we introduce the following notion of strategic robustness.  


\begin{definition}[$\alpha$(-strategically)-robust Stopping Policies]
    A stopping policy $\policy$ is $\alpha$-robust if it achieves $\alpha$-approximation to the $\prophet$ when players are strategically signaling their rewards, and it remains a $\sfrac{1}{2}$-approximation in the standard non-strategic setting. 
\end{definition}

\section{Warm-up: Characterizing Optimal Information Revealing Strategy}
\label{sec:opt infor}

We start our analysis by showing that  when the searcher adopts a threshold stopping policy (\Cref{defn:threshold}), 
each player's optimal information revealing strategy admits clean characterizations. 
\begin{proposition}[Optimal Information Revealing Strategy] 
\label{prop:optimal signaling scheme}
Given a threshold stopping policy as in \Cref{defn:threshold} with threshold $T$, 
for each player $i$: 
\squishlist
    \item if $T \le \lambda_i$, then player $i$'s optimal information revealing strategy is the no information strategy;
    \item 
    if $T > \lambda_i$, then player $i$'s optimal information revealing strategy is \emph{threshold signaling} and determined by a cutoff  $t_i$ that satisfies  
    $T = \expect{X_i|X_i \geq t_i} =  {\int_{t_i}^{\infty} x ~\mathrm{d} \prior_i(x)}/(1 - \prior_i(t_i))$. That is, player $i$'s optimal signaling scheme sends one of  two signals: $X_i \geq t_i$ or $X_i < t_i$.\footnote{In the corner case when $T$ is larger than the upper bound of player $i$'s distribution, player $i$ will never get chosen and thus they have no strategy}
\squishend
\end{proposition}
We highlight the intuition behind  \Cref{prop:optimal signaling scheme} below and defer its formal proof  to Appendix. 
Under a threshold stopping policy, every player maximizes his utility by maximizing the probability  that the signal's posterior expected reward is at least $T$ (hence is selected). When the stopping threshold $T \le \lambda_i$, this probability is $1$ and hence maximized when player $i$ simply reveals no information. When $T > \lambda_i$, this probability is maximized when player $i$ blends the highest rewards together to form a posterior mean just equal to $T$, which is exactly the scheme described in \Cref{prop:optimal signaling scheme}. 

\begin{remark}[Addressing Point Masses]
\label{rmk:point mass opt infor}
For ease of presentation, \Cref{prop:optimal signaling scheme}  assumes the distribution $\prior_i$ is continuous. 
When $\prior_i$ has point masses, \Cref{prop:optimal signaling scheme} still holds, but with a more subtle description of the pooling cutoff $t_i$. 
We provide more detailed discussions in Appendix. 
\end{remark}

\begin{remark}
\label{rmk:SPNE}
In game-theoretic terminology,  \Cref{prop:optimal signaling scheme} characterizes the  subgame Perfect Nash equilibrium (SPNE) for the multi-player sequential game induced by any  threshold stopping policy that the searcher commits to. This SPNE happens to enjoy simple structures; indeed, each player's equilibrium strategy is only a function of the threshold $T$ and not on other players' strategies or their order. This clean characterization is a consequence of the simple structure of  (static) threshold policies. If the searcher's stopping policy is instead dynamic  (i.e., allowing the decision of player $i$ to depend on previously realized rewards), SPNE is also well-defined but will adopt significantly more complex structures. While analyzing the SPNE under these dynamic policies may also be interesting, it is beyond the scope of this work since our focus is to study the power of static threshold policies, which is also a central theme in the study of prophet inequalities. Finally, we note that the optimal strategies characterized in  \Cref{prop:optimal signaling scheme} are not unique but they all result in the same utility for every player. This is because player $i$ is actually indifferent on how to disclose 
the reward when $X_i \le t_i$, leading to many different but utility-equivalent information strategies. 
\end{remark}
\xhdr{Equivalent Representation of Optimal Information Revealing Strategy} 
When $T > \lambda_i$,
  player $i$'s optimal information revealing strategy described in \Cref{prop:optimal signaling scheme} pools all the rewards $X_i \sim H_i$ above $t_i$ together to forge a conditional mean value of $T$, and pools the remaining smaller rewards into another signal. We hence also refer to $t_i$ as the pooling cutoff. This threshold signaling scheme can be equivalently represented as 
a binary-support distribution $G_i$ supported on two realizations corresponding to the two signals respectively: 
\begin{align}\label{eq:binary-reduction}
    \prob[x\sim G_i]{x = T} = 1 - \prior_i(t_i), \quad 
    \prob[x\sim G_i]{x = a_i} = \prior_i(t_i) \text{ where } 
    a_i \triangleq \frac{\lambda_i - T (1 - \prior_i(t_i))}{\prior_i(t_i)}~.
\end{align}
In the literature of information design, this distribution $G_i$ is also known as a  mean-preserving contraction of prior reward distribution $\prior_i$ \cite{GK-16b,B-51,BG-79}.
 


Viewing the player $i$'s optimal information strategy as the distribution $G_i$, one can simplify the interaction between the searcher and player $i$ 
as the following when the searcher visits player $i$: 
a random reward $X_i\sim G_i$ is realized, and the searcher stops if and only if $X_i \ge T$.
With this observation, one can also reduce the original interaction to the following simplified protocol:
(1) the searcher first decides a stopping threshold $T$;
(2) each player $i$ chooses the distribution $G_i$  as in Equation \eqref{eq:binary-reduction} according to his prior 
$\prior_i$ and the threshold $T$;\footnote{When $T\le \lambda_i$, distribution $G_i$ is a point mass at mean value $\lambda_i$.} and
(3) the searcher visits each $G_i$ sequentially and stops when the first realized reward $X_i\ge T$ 
where $X_i\sim G_i$.


We emphasize that in the above reduction, given a threshold stopping policy with threshold $T$,
the searcher's expected payoff under the strategic reward signaling can be computed as $\searcherUS(T) = \expect{X^{(q)}}$ where the expectation is over distributions $G_{1:\boxNum}$
and $X^{(q)}$ is the first realized $X_i\sim G_i$ such that $X^{(q)} \ge T$.

\section{Achieving \texorpdfstring{$\frac{1-\sfrac{1}{e}}{2}$}{}-robustness for Arbitrary Distributions}
\label{sec:general prior}

In this section, we show that for any distributions $\prior_{1:\boxNum}$,
there exists a $\frac{1-\sfrac{1}{e}}{2}$(-strategically)-robust threshold stopping policy using a threshold within the spectrum in \Cref{defn:1/2 threshold}.
The main result in this section is stated below.
\begin{theorem}
\label{thm:general prior}
For any distributions $\prior_{1:\boxNum}$,
a threshold stopping policy with threshold $T = \halfMax$ is $\frac{1-\sfrac{1}{e}}{2}$-robust.t
\end{theorem}
From \Cref{defn:1/2 threshold}, we know that using the threshold stopping policy with threshold $\halfMax$
can achieve the optimal $\sfrac{1}{2}$-approximation in the classic 
prophet inequality for any distributions $\prior_{1:\boxNum}$.
 \Cref{thm:general prior} above shows another desired property of 
the threshold $\halfMax$: it can achieve $\frac{1-\sfrac{1}{e}}{2}$-approximation even when distributions $\prior_{1:\boxNum}$ are strategically signaling their rewards, thus establishing 
its $\frac{1-\sfrac{1}{e}}{2}$-robustness.
Given the optimality of the threshold $\halfMax$ in the non-strategic setting,
it would be intriguing to ask whether this threshold $\halfMax$ can also 
achieve $\sfrac{1}{2}$-approximation under the strategic setting, or 
if there exists a threshold within the spectrum in \Cref{defn:1/2 threshold} that can achieve 
$\sfrac{1}{2}$-approximation under the strategic setting. 
Below we   show  that the answer is No. In fact,  any threshold stopping policy using a threshold from the spectrum in \Cref{defn:1/2 threshold} cannot achieve 
$(\frac{1-\sfrac{1}{e}}{2}+\varepsilon)$-approximation for any $\varepsilon >0$ under strategic reward signaling.
\begin{proposition}[Tightness of \Cref{thm:general prior}]
\label{lem:tightness general priors}
There exist distributions $\prior_{1:\boxNum}$ such that no threshold from the spectrum
in \Cref{defn:1/2 threshold} can achieve $\alpha$-robustness 
where $\alpha > \frac{1-(1 - 1/(\boxNum-1))^{\boxNum-1}}{2}$.
\end{proposition}
Notice that $\lim_{\boxNum \rightarrow\infty}  \frac{1-(1 - 1/(\boxNum-1))^{\boxNum-1}}{2} = \frac{1-\sfrac{1}{e}}{2}$. 
Thus, the above \Cref{lem:tightness general priors} establishes the tightness 
of the results in \Cref{thm:general prior}.
\begin{remark}
\label{rmk:1/2 approx general}
We point out a subtlety in the above lower bound, which leads to an intriguing open problem. 
\Cref{lem:tightness general priors} shows that any threshold within the spectrum in \Cref{defn:1/2 threshold} cannot achieve $(\frac{1-\sfrac{1}{e}}{2}+\varepsilon)$-approximation under strategic reward signaling. However, 
  this does not rule out 
the possibility of having a threshold outside that spectrum that achieves $\sfrac{1}{2}$-robustness. This is an interesting open question to resolve, though necessarily challenging since it is even already quite non-trivial to prove $\sfrac{1}{2}$-approximation in the non-strategic setting for thresholds outside the spectrum of \Cref{defn:1/2 threshold}, let alone achieving $\sfrac{1}{2}$-approximation simultaneously in both worlds.     
\end{remark}
\Cref{thm:general prior}  holds for all distributions, regardless of being discrete or continuous. 
In the remainder of this subsection, we present the proof of \Cref{thm:general prior} only for the continuous distribution case. The proof of this case carries our core ideas but is cleaner to present. We defer the proof for the discrete case, as well as the proof of \Cref{lem:tightness general priors}, to Appendix. 

\subsection{(Partial) Proof of  \Cref{thm:general prior}: The Case of Continuous Distributions}
\label{subsec: (1-1/e)/2 robust}
Our proof starts by upper bounding the prophet value $\OPT$.\footnote{
All our analysis in the main text implicitly consider the case where $\max_i \lambda_i < \sfrac{\OPT}{2}$.
If we have $\max_i \lambda_i \ge \sfrac{\OPT}{2}$, then the searcher can simply choose the threshold $T = \sfrac{\OPT}{2}$ which lies in the range of the thresholds defined in \Cref{defn:1/2 threshold} to obtain payoff $\searcherUS(T) \ge \sfrac{\OPT}{2}$. 
To see this, by \Cref{prop:optimal signaling scheme}, 
if $\max_i \lambda_i \ge \sfrac{\OPT}{2} = T$, then there exists at least one player $j$ whose $\lambda_j \ge \sfrac{\OPT}{2}$ will choose the no information revealing strategy, and the searcher surely 
obtains a payoff no smaller than $\sfrac{\OPT}{2}$.
}
\begin{lemma}[Upper Bounding $\OPT$ via Pooling Cutoffs]
\label{lem:ub opt}
Given a threshold stopping policy with threshold $T$, let 
$t_i$ be the pooling cutoff for each player $i$ 
defined as in \Cref{prop:optimal signaling scheme}. 
Let $I \triangleq \argmax_{i\in[\boxNum]} t_i$,
then we have $\OPT \le \prior_I(t_I) t_I + \sum_{i} T(1-\prior_i(t_i))$.
\end{lemma}
\begin{myproof}[Proof of \Cref{lem:ub opt}]
Let us fix a threshold $T$, and let $t_i$ be the pooling cutoff 
for player $i$ defined in \Cref{prop:optimal signaling scheme}.
Define $b_i \triangleq (X_i - t_i)^+$.
By definition, for each $i$, we have $X_i \le t_i + b_i$.
Thus, 
\begin{align*}
    \OPT = \expect[\prior_{1:\boxNum}]{\max_i X_i}
    & \le \max_i t_i + \sum_{i}\expect[\prior_i]{b_i}  \\
    & \overset{(a)}{\le} \max_i t_i + \sum_{i}(T - t_i) \cdot (1-\prior_i(t_i))\\
    & \overset{(b)}{\le} t_I \prior_I(t_I) + \sum_{i} T(1- \prior_i(t_i))~,
\end{align*}
where inequality (a) follows from 
the definition of pooling cutoff $t_i$ in \Cref{prop:optimal signaling scheme} and inequality (b) follows from the definition of $I$, $t_i \ge 0$, and rearranging the terms.
\end{myproof}
With the above upper bound of prophet value, we have the following results:
\begin{lemma}
\label{lem:middle threshold}
Let $T^\dagger$ satisfy $\prod_{i=1}^\boxNum \prior_i(t_i^\dagger) = (\frac{\boxNum-1}{\boxNum})^{\boxNum}$ where $t_i^\dagger$ is defined in
\Cref{prop:optimal signaling scheme} with threshold $T^\dagger$,
then searcher's expected payoff
$\searcherUS(T^\dagger) 
\ge  T^\dagger\left(1 - \prod_{i}\prior_i(t_i^\dagger)\right)
\ge\frac{1-(\frac{\boxNum-1}{\boxNum})^{\boxNum}}{2} \cdot \OPT$.\footnote{
Note that for non-continuous distributions $T^\dagger$ always exists but is defined a bit differently, please see Appendix
for more details.
}
\end{lemma}
\begin{myproof}[Proof of \Cref{lem:middle threshold}]
Given a stopping threshold $T$, by \Cref{prop:optimal signaling scheme}, 
we can lower bound the searcher's expected payoff as follows:
$\searcherUS(T) \ge T\cdot\left(1 - \prod_{i=1}^\boxNum \prior_i(t_i)\right)$.
Thus, together with \Cref{lem:ub opt}, we have
\begin{align*}
    \frac{\searcherUS(T)}{\OPT}
    & \ge  \frac{T\cdot\left(1 - \prod_{i=1}^\boxNum \prior_i(t_i)\right)}{\OPT}\\
    & \overset{(a)}{\ge}  \frac{T\cdot\left(1 - \prod_{i=1}^\boxNum \prior_i(t_i)\right)}{t_I \prior_I(t_I) + \sum_{i} T(1- \prior_i(t_i))}\\
    & \overset{(b)}{\ge} \frac{1 - \prod_{i=1}^\boxNum \prior_i(t_i)}{\prior_I(t_I) + \sum_{i} (1- \prior_i(t_i))}\\ 
    & \overset{(c)}{\ge}
    \frac{1 - \prod_{i=1}^{\boxNum} \prior_i(t_i)}{\boxNum + 1 - \sum_i \prior_i(t_i)} \\
    & \overset{(d)}{\ge}
    \frac{1 - \prod_{i=1}^{\boxNum} \prior_i(t_i)}{\boxNum + 1 -(\prod_{i=1}^{\boxNum} \prior_i(t_i))^{\frac{1}{\boxNum}} }~,
\end{align*}
where inequality (a) is by 
\Cref{lem:ub opt};
inequality (b) is by the fact that $t_i \le T$ for all $i\in[\boxNum]$;
inequality (c) is due to $\prior_I(t_I)\le 1$; 
and inequality (d) is by the AM-GM inequality.
Now consider the function $f(x) \triangleq \frac{1 - x}{\boxNum + 1- \boxNum x^{\frac{1}{\boxNum}}}$ over $x\in[0, 1]$.
By choosing $x^\dagger = (\frac{\boxNum-1}{\boxNum})^{\boxNum}$, we have 
$f(x^\dagger) = \frac{1 - (\frac{\boxNum-1}{\boxNum})^{\boxNum}}{2}$.
This implies that we have $T^\dagger\left(1 - \prod_{i}\prior_i(t_i^\dagger)\right)
\ge\frac{1-(\frac{\boxNum-1}{\boxNum})^{\boxNum}}{2} \cdot \OPT$.
\end{myproof}
With the above \Cref{lem:middle threshold}, we are now ready to prove \Cref{thm:general prior}:
\begin{myproof}[Proof of \Cref{thm:general prior}]
From \Cref{lem:middle threshold}, we showed
\begin{align*}
    T^\dagger\left(1 - \prod_{i}\prior_i(t_i^\dagger)\right)
    \ge \OPT\cdot \frac{1 - (\frac{\boxNum-1}{\boxNum})^{\boxNum}}{2}~,
\end{align*}
where $t_i^\dagger$ is the pooling cutoff 
defined in \Cref{prop:optimal signaling scheme} with the threshold $T^\dagger$.
By definition of $ T^\dagger$, we have $\prod_{i}\prior_i(t_i^\dagger) = (\frac{\boxNum-1}{\boxNum})^{\boxNum} \le \sfrac{1}{e}$. Thus we can deduce that $T^\dagger \ge \sfrac{\OPT}{2} = \halfMax$. 
Now let $t_i^\ddag$ be the pooling cutoff when the searcher uses the stopping threshold $\halfMax$. Then we have
\begin{align*}
    \searcherUS(\halfMax) \ge \halfMax \cdot \left(1 - \prod_{i}\prior_i(t_i^\ddag)\right) 
    \overset{(a)}{\ge}
    \halfMax \cdot \left(1 - \prod_{i}\prior_i(t_i^\dagger)\right)
    \ge \OPT\cdot \frac{1 - \sfrac{1}{e}}{2}
\end{align*}
where inequality (a) is due to $T^\dagger \ge \halfMax$, and thus 
we have $t_i^\dagger \ge t_i^\ddag$.
\end{myproof}
\section{Achieving \texorpdfstring{$\frac{1}{2}$}{}-robustness for Special Distributions}

The preceding section  showed the $\sfrac{(1-\sfrac{1}{e})}{2}$-robustness of the $\halfMax$-threshold stopping policy   for arbitrary reward distributions. 
In this section, we show that this  ratio can be improved to $\sfrac{1}{2}$-robustness when the distributions 
$\prior_{1:\boxNum}$ satisfy certain conditions: 
(1) IID distributions -- all reward distributions are identical, namely, $\prior \equiv \prior_i$ for all $i\in[\boxNum]$ (see \Cref{sec:iid prior}); and
(2) Log-concave distributions -- reward distribution $\prior_i$ has log-concave density (see \Cref{sec:log concave}). We also show that $\sfrac{1}{2}$-robustness is tight under IID distributions.

\subsection{IID Distributions}
\label{sec:iid prior}

Our main findings for IID distributions are stated below:
\begin{theorem}
\label{thm:IID prior}
For any distributions $\prior_{1:\boxNum}$ where $\prior \equiv \prior_i, \forall i\in[\boxNum]$,
a threshold stopping policy with threshold $T = \TEqual$ is $\frac{1}{2}$-robust where $T^*$ is defined in \Cref{defn:1/2 threshold}.
\end{theorem}
For IID distributions, we show that the searcher is 
able to achieve a better robustness approximation ratio
compared to arbitrary distributions. 
Below we argue that this $\sfrac{1}{2}$-robustness is tight 
in the sense that there exists no threshold
policy that can achieve better robustness approximation ratios.
\begin{proposition}[Tightness of \Cref{thm:IID prior}]
\label{prop:IID tightness}
There exist IID distributions such that there exists no threshold
stopping policy that can achieve $\alpha$-robustness where $\alpha > \frac{1}{2} + \varepsilon$ for any $\varepsilon > 0$.
\end{proposition}
We note that  \Cref{prop:IID tightness} is a slightly stronger lower bound than \Cref{lem:tightness general priors} as it rules out the possibility for \emph{all} possible threshold stopping polices, being within the spectrum in   \Cref{defn:1/2 threshold} or beyond. 
We prove \Cref{prop:IID tightness} by constructing a hard instance. In particular, we construct IID distributions with binary support. 
With this instance, we show that any threshold policy that achieves competitive ratio at least $\sfrac{1}{2}$ in non-strategic setting will have competitive ratio approaching to $0$ in the strategic setting. 
The full proof of \Cref{prop:IID tightness}, together with the proof of \Cref{thm:IID prior}, is deferred to Appendix.





A crucial requirement in our robustness study so far is that we insist that the threshold policy should, first of all, remains an $\sfrac{1}{2}$-approximation in the non-strategic setting\footnote{Note that in non-strategic setting, unlike the case for arbitrary distributions where there exists no policy that can achieve better than $\sfrac{1}{2}$-approximation, for IID distributions, $(1-\sfrac{1}{e})$-approximation can be achieved by fixed threshold together with some careful probabilistic tie-breaking rule \cite{NW-22} (or $0.7451$-approximation with an adaptive threshold policy \cite{LLPSS-21}). However, in our work, we focus on fixed threshold policy with deterministic tie-breaking rule where $\sfrac{1}{2}$-approximation is still optimal for IID distributions.}, conditioned on which we look for additional guarantee for in the strategic setting.  We conclude this section by pointing that if one was willing to give up the $\sfrac{1}{2}$-approximation in the non-strategic setting, then it is   indeed possible to have a threshold policy that   achieves better approximation (specifically, an $(1-\sfrac{1}{e})$-approximation) in the strategic setting. While this does not satisfy our robustness requirement, it is useful to note.  
\begin{corollary} 
\label{lem:IID deviation guarantee}
For any  IID distributions, 
    there exists a threshold stopping policy that is 
    $(1-\sfrac{1}{e})$-approximation  under strategic reward signaling.
    Moreover, there exist IID distributions
    such that no threshold stopping policy  
    can achieve $(1-\sfrac{1}{e} + \varepsilon)$-approximation for any $\varepsilon > 0$. 
\end{corollary}

\subsection{Log-concave Heterogeneous  Distributions}
\label{sec:log concave}
In this subsection, we show that  when the distributions $\prior_{1:\boxNum}$ satisfy certain regularity assumptions, there exist threshold stopping policies with thresholds from \Cref{defn:1/2 threshold} that can also achieve $\sfrac{1}{2}$-robustness. 
The main result in this section is stated as follows:
\begin{theorem}
\label{thm:mainOblivious}
For $\alpha,\beta > 0$, let $F_{\alpha,\beta}$ be the family of distributions with log-concave probability density functions $f$ on support $[0,1]$ such that $f(1) \ge \alpha$ and $f'(1) \ge -\beta$. If the distributions $\prior_1, \ldots, \prior_\boxNum$ are all from $F_{\alpha,\beta}$ and $\boxNum \ge 1 + \frac{\beta}{\alpha^2}$, then we always have $2 \cdot \halfMax\le \medianThresholdUB$ and
any threshold $T$ satisfying $2 \cdot \halfMax \le T\le \medianThresholdUB$ is $\sfrac{1}{2}$-robust.
\end{theorem}

A few remarks on the assumptions in \Cref{thm:mainOblivious} are worth mentioning. First, log-concavity of probability density functions\footnote{A probability density function $f: \mathbb{R} \to \mathbb{R}^+$ is log-concave if $\log(f)$ is concave.} is a commonly used assumption; they include but are not limited to: normal, beta, gamma, and exponential distributions.  The restriction to support on $[0,1]$ is for normalization reasons, hence without loss of generality. The main non-trivial restriction is that this result holds when the number of players $\boxNum$ is large enough, formally $\boxNum \ge 1 + \frac{\beta}{\alpha^2}$. This condition becomes less restrictive when    $\alpha$ (lower bounding $f(1)$) becomes larger and/or   $\beta$  (upper bounding $-f'(1)$) becomes smaller. These together, intuitively, imply that $f$ decreases slowly within $[0,1]$.


Define $\Bar{\prior}(x) \triangleq \prod_{i=1}^\boxNum \prior_i(x)$. \Cref{thm:mainOblivious} follows directly from the following \Cref{lem:convexity} and \Cref{lem:conditionForConvex}.


\begin{lemma}
\label{lem:convexity}
    If $\Bar{\prior}$ is convex, then $2\cdot \halfMax \le \medianThresholdUB$ and 
    any threshold stopping policy with the threshold $T$ satisfying $2\cdot \halfMax \le T \le \medianThresholdUB$,
    where $\halfMax, \medianThresholdUB$ are defined as in \Cref{defn:threshold}, is $\sfrac{1}{2}$-robust.
\end{lemma}

\begin{lemma}
\label{lem:conditionForConvex}
    If the distributions $\prior_1, \ldots, \prior_\boxNum$ are all from $F_{\alpha,\beta}$ and $\boxNum \ge 1 + \frac{\beta}{\alpha^2}$, then $\Bar{\prior}$ is convex.
\end{lemma}

\section{Conclusion and Future Directions}

In this paper, we study a variant of the prophet inequality problem where each random variable is associated with a strategic player who can strategically  signal their reward to the searcher. We first fully characterize the optimal information strategy of each player, then we show the threshold stopping policies that can perform robustly well under both the strategic and non-strategic settings.

Our novel consideration of natural strategic manipulations in prophet inequalities open the door for many  interesting future directions. 
First, it is interesting to see if we can
improve the $\frac{1-\sfrac{1}{e}}{2}$-robustness guarantee for arbitrary reward distributions by using not commonly used thresholds outside the spectrum of \Cref{defn:1/2 threshold}. 
This may be   technically challenging  since finding a threshold  outside   this spectrum  with optimal  $\sfrac{1}{2}$-approximation for the classic non-strategic prophet inequality is already non-trivial.  
Another interesting direction is to go beyond   threshold policies. That leads to a different research theme, not about robustness of the classic prophet inequality, but rather in the search of potentially much more complex best-of-both-world policies. 
Our preliminary result reveals that a dynamic threshold stopping policy (using different thresholds for different players) has the potential to help the searcher to achieve more than $\frac{1-\sfrac{1}{e}}{2}$-approximation under the strategic setting.
However, it is still unclear whether $\sfrac{1}{2}$-robustness is achievable under threshold stopping policies, no matter whether it is static threshold or dynamic threshold. 
Finally, even if one only cares about the performance under the strategic reward signaling (and ignore the non-strategic world), it is still unclear which static threshold stopping policy achieves the highest competitive ratio. The authors in \citep{DMZ-22} study a model with a very similar mathematical structure and proved that this best upper bound is strictly less than $\sfrac{1}{2}$ in Section 4.1, but an interesting direction is improving this bound further.
We note that all our constructed examples in the hardness results (e.g., \Cref{lem:tightness general priors} and \Cref{prop:IID tightness}) 
do not rule out the existence of such $\sfrac{1}{2}$-approximation threshold stopping policies.
 

\vspace{-2mm}
\paragraph{Acknowledgment. } Haifeng Xu is supported in part by    the NSF Award CCF-2303372,   Army Research Office Award W911NF-23-1-0030 and ONR Award N00014-23-
1-2802.


\bibliographystyle{plainnat}
\bibliography{mybib}

\newpage
\appendix
\section{Missing Proofs in \texorpdfstring{\Cref{sec:opt infor}}{}}
\label{apx:proof opt infor}

\begin{proof}[Proof of \Cref{prop:optimal signaling scheme}]
    Note there always exists a cutoff $t_i$ such that $T = \expect[x\sim \prior_i]{x \mid x\ge t_i}$, because $\prior_i$ is continuous and $\expect[x\sim \prior_i]{x \mid x\ge t_i}$ is increasing with respect to $t_i$. 
    Let $G_i$ be the binary-support distribution defined in the Equivalent Representation of Optimal Information Revealing Strategy (when $T \le \lambda_i$, $G_i$ is a point mass on value $\lambda_i$).
    Then we know that player $i$'s probability of being accepted 
    when the searcher reaches player $i$ is given by 
    $p_i \triangleq 1 - \prior_i(t_i)$.
    
    We first argue that this probability is the maximum probability 
    of player $i$ being accepted conditional on the searcher reaching 
    player $i$. 
    To see this, notice that once the searcher reaches player $i$, 
    the decision of the searcher on whether to accept or reject player $i$ only depends on the searcher's expected posterior reward (i.e., whether the expected posterior reward is larger or smaller than the threshold). 
    With this observation, it is known that every signaling scheme can be modeled as a mean-preserving contraction (MPC hereafter) $G\in \Delta(\supp(\prior_i))$ of the original reward distribution $\prior_i$ \cite{GK-16b,KMZL-17}.
    By definition \cite{B-51,BG-79}, we know that a distribution $G\in\Delta(\supp(\prior_i))$ is an MPC of $\prior_i$ if and only if the following holds 
    \begin{align}
        \label{eq:MPC}
        \int^t_{0} G(x) \mathrm{d}x \le \int^t_{0} \prior_i(x) \mathrm{d} x, \quad \forall t\in\supp(\prior_i); \quad
        \text{and} \quad 
        \expect[x\sim G]{x} = \lambda_i = \expect[x\sim \prior_i]{x}~.
    \end{align}
    Notice that for any MPC $G$, the probability that the searcher 
    accepts player $i$ is given by $\prob[x\sim G]{x\ge T}$.
    Together with the above MPC definition \Cref{eq:MPC} and 
    the construction of the distribution $G_i$, it is easy to see that 
    we always have $p_i \ge \prob[x\sim G]{x\ge T}$ for any MPC $G$.

    We next argue that $G_{1:\boxNum}$ defined in \Cref{prop:optimal signaling scheme} 
    indeed forms a subgame perfect 
    Nash equilibrium among the players' game. 
    We prove this by induction. Looking at the first player, 
    his payoff exactly equals to his probability on being accepted by the 
    searcher, which is maximized when he adopts the information strategy $G_1$. 
    Conditional on player $1$ using the strategy $G_1$, player $2$'s expected payoff on using an MPC $G$ exactly equals to $\prob[x\sim G]{x\ge T}\cdot \prior_1(t_1)$. As we have argued earlier, $\prob[x\sim G]{x\ge T}$ is maximized when $G = G_2$. 
    Similarly arguments can be iteratively applied to every player $i$. Thus
    we have shown that $G_{1:\boxNum}$ is indeed a subgame perfect 
    Nash equilibrium.
%
\end{proof}

\begin{proposition}[Optimal Information Revealing Strategy -- Point Mass ]
\label{prop:optimal signaling scheme point mass}
Let $s = \max \supp(X_i)$, if it exists. Given a threshold stopping policy defined in \Cref{defn:threshold} with threshold $T$, 
for each player $i$, we define a \textit{rejection probability} $p_i$: 
\squishlist
    \item if $T \le \expect{X_i}$, then player $i$'s optimal information revealing strategy is the no information revealing strategy, so we define $p_i = 0$;
    \item 
    if $s \ge T > \expect{X_i}$, let $t_i$ and $q_i$ satisfy 
    $T \cdot (1 - \prior_i(t_i^+) + q_i) = {\int_{t_i^+}^{\infty} x ~\mathrm{d} \prior_i(x)} + q_it_i$,
    then player $i$'s optimal information revealing strategy generates two signals: 
    \squishlist
    \item[(a)] with probability $1-p_i \triangleq 1 - \prior_i(t_i^+) + q_i$,
    one signal leads to the searcher's posterior mean  $T$;
    \item[(b)] with probability $p_i = \prior_i(t_i^+) - q_i$,
    one signal leads to the searcher's posterior mean $\frac{\expect{X_i} - T (1 - \prior_i(t_i^+) + q_i)}{\prior_i(t_i^+) - q_i}$.
    \squishend
    such that $p_i$ is minimized;
    \item if $T > s$, then player $i$ cannot use any information revealing strategy to get a positive payoff, so we define $p_i = 1$
\squishend
\end{proposition}

The player's optimal information revealing strategy in this case is very similar to the case where the prior is continuous; in that case $q_i$ always equals $0$ and $t_i$ always exists and is unique. Here there can be many different $t_i$ when $q_i = 0$. However, $p_i$ is always unique, as each player wishes to minimize $p_i$ to maximize the probability $(1-p_i)$ of getting selected on his turn.

Also, we could have $q_i > 0$ for a unique $t_i$ when there is a point mass on $t_i$. This can be interpreted as \textit{percentage threshold pooling}, where a mass of $q_i$ on $t_i$ is pooled to $T$, and the remaining mass is pooled towards the other posterior mean. This idea is explored deeply in \cite{CFHOV-17} for IID distributions. Finally, if $T > s$, then there is obviously no way for the player to be accepted by the searcher. 

The proof for this proposition is the exact same approach in the continuous priors case, with details on percentage threshold pooling, and hence the exact details can be left to the reader.

Finally, we have the following remark synonymous to the continuous case:

\begin{remark}[Equivalent Representation of Optimal Information Revealing Strategy -- Point Mass]
We note that when $\expect{X_i} \le T \le s$,
the player $i$'s optimal information revealing strategy defined in \Cref{prop:optimal signaling scheme point mass} can be effectively represented as 
a distribution $G_i$ supported on only two realizations, with $p_i$ defined in \Cref{prop:optimal signaling scheme point mass}:
\begin{align*}
    \prob[x\sim G_i]{x = T} = 1 -p_i, \quad 
    \prob[x\sim G_i]{x = a_i} = p_i\text{ where } 
    a_i \triangleq \frac{\expect{X_i} - T (1 - p_i)}{p_i}~.
\end{align*}
\end{remark}

\begin{lemma} [Inducing $T,t_i$ via $p_i$]
\label{lem:inducing T via p}
For any $0 \le p_i \le 1$, there exists a unique pair $(T,t_i)$ such that there exists $q_i \ge 0$ such that $T(1-p_i) = \int_{t_i^+}^{\infty} x\mathrm{d}\prior_i(x) + q_it_i$.
\end{lemma}

The intuition behind this lemma's proof is that in the continuous case, we can always find a pooling cutoff $t_i$ such that $p_i = \int_0^{t_i} \mathrm{d}\prior_i(X_i)$, and in the non-continuous case, this pooling trick should also always be possible if we allow for partial pooling on point masses on $t_i$. Also, note that $T$ is clearly increasing with respect to $p_i$.

Like in the main text, we refer to $t_i$ as the pooling cutoff.

With this lemma, we have the following natural definition:

\begin{definition}[Inducing $T,t_i$ via $p_i$]
    We say that a rejection probability $0 \le p_i < 1$ induces threshold $T$ and pooling cutoff $t_i$ if $T$ and $t_i$ satisfy \Cref{lem:inducing T via p}. 
\end{definition}

\begin{lemma} [Inducing $T$ via $P$]
\label{lem:inducing T via P}
For any $0 \le P < 1$, there exists a $T$ such that
there exists a set of pairs $\{(t_i,q_i)\}$ and a set $\{p_i\}$ such that $T(1-p_i) = \int_{t_i^+}^{\infty} x\mathrm{d}\prior_i(x) + q_it_i$ for all $i$ and $\prod_{i=1}^{\boxNum} p_i = P$.
\end{lemma}

\begin{proof}
    For any $T$, let $\underline{P}(T)$ and $\overline{P}(T)$ be the minimum and maximum values of $\prod_{i=1}^{\boxNum} p_i$ among all sets $\{p_i\}$ such that $p_i$ induces $T$. Note that $\underline{P}(\max_i \priorMean_i) = 0$ since for $\priorMean_j = \max_i \priorMean_i$, $p_j = 0$, and also for any $\epsilon > 0$, there exists a large enough $T'$ such that $\overline{P}(T') = 1-\epsilon$. Note that $\underline{P}(T) \le \overline{P}(T)$ and both functions are increasing with respect to $T$ since $T$ is increasing with respect to $p_i$. Hence, by the Intermediate Value Theorem, for any $P$, we should find a $T$ such that $\underline{P}(T) \le P \le \overline{P}(T)$.
\end{proof}

So with this lemma, we have a definition:

\begin{definition}[Inducing $T$ via $P$]
    We say that a probability $0 \le P \le 1$ induces a threshold $T$ if $T$ satisfies \Cref{lem:inducing T via P}. 
\end{definition}

Note there can be multiple such $T$ induced by any given $P$.

\section{Missing Proofs in \texorpdfstring{\Cref{sec:general prior}}{}}
\label{apx:proof general prior}

We use the following example to prove \Cref{lem:tightness general priors}.
\begin{example}[Example for Proving \Cref{lem:tightness general priors}]
   Let there be $\boxNum$ distributions such that $\prior_1 = \Delta(d)$ where $d = \boxNum-1-\epsilon$ for $\epsilon > 0$, and for $i \neq 1$, the distributions $\prior_i$ have mean $1$ and support on $\{0,s\}$, where $s$ is significantly large.
\end{example}

\begin{proof}[Proof of \Cref{lem:tightness general priors}]
    Firstly, since the binary distributions $\prior_i$ for $i \neq 1$ have mass $\frac{1}{s}$ on $s$ and $1-\frac{1}{s}$ on $0$, we can compute
    \begin{align*}
    \OPT = d \cdot \left(1-\frac{1}{s}\right)^{\boxNum-1} + s \cdot \left(1-\left(1-\frac{1}{s}\right)\right)^{\boxNum-1}
    \end{align*}
    Note that $\lim_{s \to \infty} s \cdot (1 - (1-\frac{1}{s}))^{\boxNum-1} = \boxNum-1$, so we can make $s$ large enough (say $s > S(d)$) for some function $S$ in terms of $d$ such that $\OPT/2 > d$. 
    
    Furthermore, one can compute $\medianThresholdUB = d$, since $Pr[\max\{X_i\} \ge d] = (1-\frac{1}{s})^{\boxNum-1} > 1/2$ for large enough $s$ and $Pr[\max\{X_i\} \ge T] = 0$ for any $T < d$. Also, it is well-known from \cite{KW-12} that $\TEqual \ge \OPT/2 > d$, so one can compute $\TEqual = (\boxNum-1) \cdot \frac{1}{s} \cdot (s - \TEqual)$ and thus $\TEqual = \frac{\boxNum-1}{1+(\boxNum-1)/s} < \boxNum-1$.
    
    We know for $T > d$, $\searcherUS(T) = T \cdot (1 - (1-\frac{1}{T})^{\boxNum-1})$, since players $i > 1$ use optimal information revealing strategies $G_i$ with binary support $\{0,T\}$ and mass $1/T$ on $T$ in the strategic setting. This payoff is increasing and thus maximized within the spectrum at $T = \TEqual$, since $\TEqual \ge \max\{d,\OPT/2\}$. Thus for any $T$ within the spectrum,
    \begin{align*}
        \searcherUS(T) \le \searcherUS(\TEqual) < \searcherUS(\boxNum-1) \le (\boxNum-1) \cdot \left(1-\left(1 - \frac{1}{\boxNum-1}\right)^{\boxNum-1}\right)
    \end{align*}
    
    Also, one can compute $\lim_{s \to \infty, d \to \boxNum-1, s \ge S(d)} \OPT = 2(\boxNum-1)$. Therefore, we can conclude 
    \begin{align*}
        \frac{\searcherUS(T)}{\OPT} < \frac{1-\left(1 - \frac{1}{\boxNum-1}\right)^{\boxNum-1}}{2} + \epsilon'
    \end{align*}
    for any $\epsilon' > 0$ where $s$ and $\epsilon$ are made large and small enough, respectively.
\end{proof}

\begin{remark}[Continuation of \Cref{rmk:1/2 approx general}]
    For this example, one can show that $T = s$ is a $1/2$-robust approximation threshold stopping policy. This is because under the strategic setting, the players' optimal information revealing strategies $G_i$ for $i \neq 1$ are binary support distributions with $\supp(G_i) = \{0,s\}$, which are equivalent to the original distributions $\prior_i$. Then the searcher's payoff in both the strategic and non-strategic settings is $\searcherUS(s) = s \cdot (1 - (1-\frac{1}{s}))^{\boxNum-1} > \OPT/2$.
\end{remark}

\begin{lemma}[Upper Bounding the Prophet Value via the Pooling Thresholds -- Point Mass]
\label{lem:ub opt point mass}
Given a probability $0 \le P \le 1$, let 
the common threshold $T$ be induced by $P$, and values $p_i$, $t_i$ and $q_i$ for each player $i$ 
be defined as in \Cref{lem:inducing T via P}. 
Let $I \triangleq \argmax_{i\in[\boxNum]} t_i$,
then we have 
\begin{align*}
    \OPT \le p_I t_I + \sum_{i} T(1-p_i)~.
\end{align*}
\end{lemma}
\begin{proof}[Proof of \Cref{lem:ub opt point mass}]
Define $b_i \triangleq (X_i - t_i)^+$.
By definition, for each $i$, we have $X_i \le t_i + b_i$.
Thus, 
\begin{align*}
    \OPT = \expect[\prior_{1:\boxNum}]{\max_i X_i}
    & \le \expect[\prior_{1:\boxNum}]{\max_i (t_i + b_i)}\\
    & \le \max_i t_i + \sum_{i}\expect[\prior_i]{b_i}  \\
    & \le \max_i t_i + \sum_{i}\int_{t_i^-}^{\infty} (X_i-t_i)\mathrm{d}\prior_i(X_i)  \\
    & = \max_i t_i + \sum_{i}\left[\int_{t_i^+}^{\infty} (X_i-t_i)\mathrm{d}\prior_i(X_i) + q_i(t_i-t_i)\right]  \\
    & \overset{(a)}{\le} \max_i t_i + \sum_{i}(T - t_i) \cdot (1-\prior_i(t_i^+) + q_i)\\
    & \overset{(b)}{\le} t_I p_I + \sum_{i} T(1- p_i)~,
\end{align*}
where inequality (a) follows from 
the definition of pooling threshold $t_i$ in \Cref{lem:inducing T via p}, and inequality (b) follows from the definitions of $I$, $p_i$, $q_i$, $t_i \ge 0$, and rearranging the terms.
\end{proof}

\begin{lemma}[Lower Bounding the Searcher's Payoff via the Pooling Thresholds -- Point Mass]
\label{lem:lb payoff point mass}
Given a probability $0 \le P < 1$, let $T$ be a threshold induced by $P$ and the rejection probabilities $p_i$ be defined by the player's optimal information revealing strategy from \Cref{prop:optimal signaling scheme point mass}. Then the player's payoff $\searcherUS(T) = T(1-\prod_{i=1}^{\boxNum} p_i) \ge T(1-P)$.
\end{lemma}

\begin{proof}
    Using notation from the proof of \Cref{lem:inducing T via P}, our $p_i$ induces $T$, and must also satisfy $\prod_{i=1}^{\boxNum} p_i = \underline{P}(T)$. This is because when the players see threshold $T$, each player wishes to minimize their rejection probability $p_i$, and $\underline{P}(T)$ is the minimum such product of probabilities. Hence $\underline{P}(T) \le P$. The searcher's payoff is also $T$ times the conditional probability at least one player reveals signal $T$, which is $T(1-\prod_{i=1}^{\boxNum} p_i)$.
\end{proof}

\begin{lemma}
\label{lem: middle threshold point mass}
Let 
$T = T^\dagger$ be induced by $P = (\frac{\boxNum-1}{\boxNum})^{\boxNum}$, then we have the searcher's expected payoff
$\searcherUS(T^\dagger) \ge  T^\dagger\left(1 - \prod_{i}\prior_i(t_i^\dagger)\right)
    \ge\frac{1-(\frac{\boxNum-1}{\boxNum})^{\boxNum}}{2} \cdot \OPT$.
\end{lemma}
\begin{proof}[Proof of \Cref{lem: middle threshold point mass}]
To see the above result,  notice that one can induce a threshold $T^\dagger$ via $P = (\frac{\boxNum-1}{\boxNum})^{\boxNum}$, and replacing $p_i$ with $p_i'$ in \Cref{lem:ub opt point mass} and \Cref{lem:lb payoff point mass} where $\prod_{i=1}^{\boxNum} p_i' = P$, we can see that $\frac{\searcherUS(T^\dagger)}{\OPT} \ge \frac{T^\dagger(1-\prod_{i=1}^{\boxNum} p_i')}{p_I' t_I + \sum_{i} T^\dagger(1-p_i')}$, and continue the algebraic proof of \Cref{lem:middle threshold} in the same way but replace $\prior_i(t_i)$ by $p_i'$.
\end{proof}

With the above \Cref{lem: middle threshold point mass}, we are now ready to prove \Cref{thm:general prior}:
\begin{myproof}[Proof of \Cref{thm:general prior}]
We notice that from the proof of \Cref{lem: middle threshold point mass}, we know
\begin{align*}
    T^\dagger\left(1 - \prod_{i}\prior_i(t_i^\dagger)\right)
    \ge \OPT\cdot \frac{1 - (\frac{\boxNum-1}{\boxNum})^{\boxNum}}{2}~,
\end{align*}
where $t_i^\dagger$ is the pooling cutoff 
defined in \Cref{prop:optimal signaling scheme} with the threshold $T^\dagger$.
By definition of $ T^\dagger$, we have $\prod_{i}\prior_i(t_i^\dagger) \le \prod_{i} p_i' = (\frac{\boxNum-1}{\boxNum})^{\boxNum} \le \sfrac{1}{e}$ using the same definition of $p_i'$ in the proof. Thus we can deduce that $T^\dagger \ge \sfrac{\OPT}{2} = \halfMax$. 
Now we can finish in the same manner as in the main text.
\end{myproof}

\section{Missing Proofs in \texorpdfstring{\Cref{sec:iid prior}}{}}
\label{apx:proof iid prior}

In below, we present the proof of \Cref{thm:IID prior} for the continuous distribution case.
The proof for the distributions that have point masses follows similarly with a subtlety mentioned below:
\begin{remark} [\Cref{thm:IID prior} -- Point Mass]
    \label{thm:IID prior point mass}
    The only subtlety in the main text proof with continuity of distribution $\prior$ is defining $p$ and $t^*$. For noncontinuous cases, we can let $p$ and $t^*$ in the main text proof be defined as $p_i$ and $t_i$ in \Cref{prop:optimal signaling scheme point mass} and \Cref{lem:inducing T via p} respectively, using threshold $T = \TEqual$. Since $p \le q$ still, the rest of the proof still follows.
\end{remark}

\begin{myproof}[Proof of \Cref{thm:IID prior}]
Let $\prior \equiv \prior_i$ for all $i\in[\boxNum]$.
Let the threshold $\TEqual$ satisfy $\TEqual = \sum_i\expect[\prior_i]{(X_i-\TEqual)^+}$.
Let
$t^*\equiv t_i^*$ for all $i\in[\boxNum]$ 
be the corresponding pooling cutoff defined in \Cref{prop:optimal signaling scheme} with the threshold $\TEqual$.
For notation simplicity, let $p \triangleq \prior(t^*)$ when $\prior$ is continuous.
Notice that by definition, we have
$\TEqual (1-p) = t^* (1-p) + \expect[\prior]{b^*}$ where $b^* \triangleq (X - t^*)^+$ for $X\sim \prior$, and let 
$c^* \triangleq  (X - T^*)^+$ for $X\sim \prior$.
By definition, we know that $\TEqual = \sum_{i} \expect[\prior_i]{(X_i-\TEqual)^+} = \boxNum \expect[\prior]{c^*}$.

From the searcher's expected payoff we know that
\begin{align*}
    \searcherUS(\TEqual) 
    = \TEqual \cdot (1-p^\boxNum) 
    & = \TEqual \cdot (1-p) \cdot (1+\cdots+p^{\boxNum-1}) \\
    & = (t^*(1-p) + \expect[\prior]{b^*}) \cdot \frac{1-p^\boxNum}{1-p} \\
    & \ge \expect[\prior]{b^*} \cdot \frac{1-p^\boxNum}{1-p} 
    \ge \expect[\prior]{c^*} \cdot \frac{1-p^\boxNum}{1-p}
    = \frac{\TEqual}{\boxNum} \cdot \frac{1-p^\boxNum}{1-p}~.
\end{align*}
Notice that the above inequalities also imply 
$p\le 1 - \sfrac{1}{\boxNum}$.

We now upper bound the prophet value $\OPT$.
Let $q \triangleq \prob[\prior]{X < \TEqual}$ 
and let $c_i \triangleq (X_i - \TEqual)^+$.
Clearly, we have $q \ge p$. 
Then 
\begin{align*}
    \OPT 
    & \le \expect[\prior_{1:\boxNum}]{\max_i (\TEqual + (X_i - \TEqual)^+)} \\
    & \le \TEqual + \expect[\prior_{1:\boxNum}]{\max_i (X_i - \TEqual)^+}\\
    & \le \TEqual + \sum_{j=1}^\boxNum 
    \prob[\prior_{1:\boxNum}]{|\{i: X_i \ge \TEqual\}|=j} 
    \cdot j\expect[\prior]{c^*} \\
    & \le \TEqual + \sum_{j=1}^\boxNum \binom{\boxNum}{j} (1-q)^j q^{\boxNum-j} \expect[\prior]{c^*} \\
    & = \TEqual + \boxNum(1-q)\expect[\prior]{c^*} 
    = \TEqual (2-q) 
    \le \TEqual (2-p)~.
\end{align*}
Thus, we have
\begin{align*}
    \frac{\searcherUS(\TEqual) }{\OPT }
    \ge \frac{\TEqual (1-p^\boxNum)}{\TEqual \cdot (2-p)}
    = \frac{1 - p^\boxNum}{2-p} 
    \overset{(a)}{\ge} \frac{1}{2}~,
\end{align*}
where inequality (a) holds as 
function $f(x) \triangleq \frac{1-x^\boxNum}{2-x} \ge \frac{1}{2}$ for all $x\in[0, 1-\sfrac{1}{\boxNum}]$ for any $\boxNum\ge 2$.
%
\end{myproof}

\begin{proof}[Proof of \Cref{prop:IID tightness}]
\label{prop:IID tightness proof}

Consider an instance with $N$ iid distributions on binary support $\{N-\alpha_1, N + \alpha_2\}$ for $ \alpha_1 \in (0, 1), \alpha_2 > 0$   such that the probabilities are chosen to make the   mean of each distribution   precisely $N$ --- i.e., the probability of taking the smaller value is $\frac{\alpha_2}{\alpha_1 + \alpha_2}$. For the purpose of this proof, one should think of $N, \alpha_2$ as being very large. The prophet value $\OPT$ can be calculated as following:  
\begin{align*}
    \OPT = (N +  \alpha_2) \cdot  \left[1 - \left(\frac{\alpha_2}{\alpha_2+\alpha_1}\right)^N\right]  + (N - \alpha_1) \cdot \left(\frac{\alpha_2}{\alpha_2 + \alpha_1}\right)^N.
\end{align*} 

We first claim that, for \emph{any} $\alpha_1 \in(0, 1)$, no threshold $T \in (N- \alpha_1,  N + \alpha_2]$ can guarantee at least $1/2$ competitive ratio  for non-strategic settings. Intuitively, this is because any such threshold will give up any realization of $N- \alpha_1$ value, which leads to bad performance when the probability of having at least one realization of  $ N + \alpha_2$ is small (i.e., when $\alpha_2$ is very large). Detailed calculation is as below:   any threshold $T \in (N- \alpha_1 , N + \alpha_2] $ will only accepts realized $N + \alpha_2$ value, hence the expected payoff for non-strategic (ns) setting is (i.e., the first term of $\OPT$) 
\begin{align*}
    \searcherUNS(T) = (N + \alpha_2)  \cdot \left[ 1 - \left(\frac{\alpha_2 }{\alpha_2 + \alpha_1 }\right)^N\right].
\end{align*}

Therefore,
\begin{eqnarray*}
    \frac{\searcherUNS(T)}{\OPT}  &=&  \frac{ (N + \alpha_2)  \cdot [ 1 - (\frac{\alpha_2 }{\alpha_2 + \alpha_1 })^N]  }{ (N + \alpha_2)  \cdot [ 1 - (\frac{\alpha_2 }{\alpha_2 + \alpha_1 })^N]  + (N - \alpha_1) \cdot (\frac{ \alpha_2}{\alpha_2 + \alpha_1})^N  }  \\ 
    &=&  \frac{  1  }{ 1+ \frac{N - \alpha_1}{N + \alpha_2} \cdot \frac{\left(\frac{ \alpha_2}{\alpha_2 + \alpha_1}\right)^N}{   1 - \left(\frac{\alpha_2 }{\alpha_2 + \alpha_1 }\right)^N}} ~.
    \end{eqnarray*} 
We now analyze the limit of the key term $    \frac{N - \alpha_1}{N + \alpha_2} \cdot \frac{(  \frac{ \alpha_2}{\alpha_2 + \alpha_1})^N}{   1 - (\frac{\alpha_2 }{\alpha_2 + \alpha_1 })^N   }   $. To ease the analysis, we re-parameterize the above term with $\alpha_2 = \frac{\alpha_1 N }{ \alpha}$ with free parameter $\alpha$ to obtain 
\begin{eqnarray*}
  \frac{N - \alpha_1}{N + \alpha_2} \cdot \frac{( \frac{ \alpha_2}{\alpha_2 + \alpha_1})^N}{   1 - (\frac{\alpha_2 }{\alpha_2 + \alpha_1 })^N   }     =     \frac{N - \alpha_1}{N} \frac{1}{1 + \alpha_1/\alpha}  \cdot \frac{(  \frac{ N }{N + \alpha})^N}{   1 - (\frac{ N }{N + \alpha })^N   }  = A(N, \alpha; \alpha_1)~.
\end{eqnarray*} 
Since $\alpha_1 \in (0, 1)$ is fixed. Now fix $\alpha$ and let $N \to \infty$, we have 
\begin{align*}
    \lim_{N \to \infty} A(N, \alpha; \alpha_1) =    \frac{1}{1 + \alpha_1/\alpha} \cdot \frac{e^{-\alpha} }{   1 - e^{-\alpha}  }~.
\end{align*}
Now, letting $\alpha \to 0$ and applying standard l'hopital's rule, we have 
\begin{align*}
    \lim_{ \alpha \to 0 }   \frac{1}{1 + \alpha_1/\alpha} \cdot \frac{e^{-\alpha} }{   1 - e^{-\alpha}  }  =  \lim_{ \alpha \to 0 }   \frac{\alpha e^{-\alpha} }{(\alpha + \alpha_1)(1 - e^{-\alpha} ) }  = \frac{1}{\alpha_1}~.
\end{align*} 
As a consequence, we have
\begin{align*}
    \lim_{N \to \infty, \alpha_2 \to \infty}   \frac{\searcherUNS(T)}{\OPT} = \frac{\alpha_1}{\alpha_1  + 1}~.
\end{align*} Thus, for any $\alpha_1 \in (0, 1)$, there exists an IID instance with large enough $N$ and $ \alpha_2 $ such that no threshold $T \in (N-\alpha_1, N +  \alpha_2]$ can achieve CR at least $1/2$.  

The above analysis shows that the only threshold $T$ that can guarantee a CR at least $1/2$ in the above class of instances must satisfy $T \le N-\alpha_1$. 
Any such threshold will achieve searcher's utility $\searcherUS(T) = N$ in the strategic setting since all players will reveal no reward information, leading to expected reward value $N$ for each player. 
A similar limit analysis using the same  reparameterization $\alpha_2 = \frac{\alpha_1 N }{ \alpha}$ shows that  $\lim_{\alpha \to 0 }\lim_{N  \to \infty}    \frac{\searcherUS(T)}{\OPT} = \frac{1}{1+\alpha_1}$.
The arguments above shows that for any $\alpha_1 < 1$, as $N, \alpha_2 \to \infty$ in the above instances,  any threshold that has CR at least $1/2$ in non-strategic setting will have CR tending to $ \frac{1}{1+\alpha_1}$ in the strategic setting. Since this conclusion holds for any $\alpha_1 \in (0, 1)$, it rules out the possibility of  having an $(1/2+\epsilon)$-robust threshold (even outside the $1/2$-approximation threshold spectrum as in Definition 2.2) for any  $\epsilon > 0$. 
\end{proof}

We now prove \Cref{lem:IID deviation guarantee} using standard techniques within our paper.

\begin{proof}[Proof of \Cref{lem:IID deviation guarantee} for continuous priors]
\label{proof: IID deviation guarantee}
    For threshold $T$ under strategic information revealing, let $p_i\triangleq \prior_i(t_i)$ where $t_i$ is the pooling cutoff defined in
    \Cref{prop:optimal signaling scheme}.
    Then we know $T(1-p_i) = t_i(1-p_i) + \mathbb{E}[b_i]$, where $b_i = (X_i - t_i)^+$. By \Cref{lem:ub opt}, 
    \begin{align*}
        \OPT \le t_i + \boxNum\sum_i \mathbb{E}[b_i]
    \end{align*}
    
    Also, one can compute
    \begin{align*}
        \searcherUS(T) = T \cdot \left(1-p_i^\boxNum\right) = T\cdot \left(1-p_i\right)\cdot \left(1 + \cdots + p_i^{\boxNum-1}\right)
    \end{align*}
    
    Now pick $T$ such that $p_i = 1-\frac{1}{\boxNum}$. Then 
    \begin{align*}
        \searcherUS(T) &= \left(t_i(1-p_i) + \mathbb{E}[b_i]\right) \cdot \left(1 + \cdots + p_i^{\boxNum-1}\right) \\
        &= \left(\frac{t_i}{\boxNum} + \mathbb{E}[b_i]\right) \cdot \frac{1-p_i^{\boxNum}}{1-p_i} \\ &\ge \frac{\OPT}{\boxNum} \cdot\left(1-\left(1-\frac{1}{\boxNum}\right)\right)^{\boxNum}\cdot \boxNum \\ &\ge \left(1-\frac{1}{e}\right) \cdot \OPT
    \end{align*}
\end{proof}

Now for noncontinuous priors, we can induce a threshold $T$ by $P = (1-\frac{1}{\boxNum})^{\boxNum}$ and refer to \Cref{lem: middle threshold point mass} for details on an explanation of why this $T$ works.

\begin{lemma}[Tightness of \Cref{lem:IID deviation guarantee}]
\label{lem:tightness IID deviation guarantee}
Under strategtic reward signaling, there exist IID distributions $\prior_{1:\boxNum}$ such that no threshold stopping policy
can achieve $(1-\frac{1}{e} + \epsilon)$-approximation for any $\epsilon > 0$.
\end{lemma}

\begin{proof}
    In \cite{CFHOV-17}, the authors gave an example of distributions that cannot exceed $1-\frac{1}{e}$-approximation if \textit{percentage threshold} policies in the non-strategic setting are allowed; that is, if one accepts any signal $X_i > T$ and partially accepts signal $T$ with a designated probability. The IID distributions they gave have the following form: $\boxNum = n^2$, and each distribution is supported only on $3$ points: 
    \begin{itemize}
    \item $\frac{n}{e-2}$ with probability $n^{-3}$,
    \item $1$ with probability $n^{-1}$
    \item $0$ otherwise
    \end{itemize}

    From these distributions, one can compute $\lim_{n \to \infty} \OPT = \frac{e-1}{e-2}$, and the optimal percentage threshold policy is accepting any signal greater than $1$, and accepting signal $1$ with probability $n^{-1}$ conditioning on seeing the signal as $1$ as $n \to \infty$ (the optimal probability is roughly $n^{-1}$ for large enough $n$). This policy gives the searcher a payoff of approximately $(\frac{e-1}{e-2}) \cdot (1+\frac{1}{e})$. Note that any threshold policy in the strategic setting can be derived from a percentage-based threshold policy using \Cref{prop:optimal signaling scheme point mass}; here we set $t_i = 1$ and $q_i = n^{-2}$ so that
    \begin{align*}
        T = \frac{\frac{n}{e-2} \times n^{-3} + 1 \times n^{-2}}{n^{-3} + n^{-2}} = \frac{(e-1)n}{(e-2)(n+1)}
    \end{align*} These two policies give the searcher the same payoff. It follows that, for $n$ large enough, there is no threshold policy in the strategic setting that gives the searcher payoff higher than $(\frac{e-1}{e-2}) \cdot (1+\frac{1}{e}) + \epsilon$ for any $\epsilon > 0$.
\end{proof}

\section{Missing Proofs in \texorpdfstring{\Cref{sec:log concave}}{}}
\label{apx:proof log concave}

\begin{proof}[Proof of \Cref{lem:convexity}]
    When $\Bar{\prior}$ is convex, $ 2 \cdot \halfMax = \priorMean(\Bar{\prior}) \le \medianThresholdUB(\Bar{\prior})$. 
    To see this, we first note that by Jensen's inequality, we must have
    \begin{align*}
        \Bar{\prior}(\lambda(\Bar{\prior})) \le \sfrac{1}{2}
    \end{align*}
    Moreover by definition of $\medianThresholdUB$, we know:
    \begin{align*}
        \prob[\Bar{\prior}]{X \ge \medianThresholdUB(\Bar{\prior})} \ge \sfrac{1}{2} \ge 
        \Bar{\prior}(\lambda(\Bar{\prior}))~,
    \end{align*}
    implying $\priorMean(\Bar{\prior}) \le \medianThresholdUB(\Bar{\prior})$.

    Thus, for any $T$ such that $ \priorMean(\Bar{\prior}) \le T\le \medianThresholdUB(\Bar{\prior})$, we have $\Bar{\prior}(T) \le \frac{1}{2}$, so 
    \begin{align*}
        (1-\Bar{H}(T))\cdot T \ge \priorMean(\Bar{\prior})/2
    \end{align*}
    The conclusion follows by \Cref{{lem:robustSimpleThres}}, which is stated and proved below.
\end{proof}

\begin{lemma}
\label{lem:robustSimpleThres}
Given any threshold stopping policy with threshold $T$ such that  $ \priorMean(\Bar{\prior}) \le T\le \medianThresholdUB(\Bar{\prior})$, if $(1-\Bar{H}(T))T \ge \alpha \lambda(\Bar{\prior})$, then this policy is $\alpha$-robust. 
\end{lemma}
\begin{proof}[Proof of \Cref{lem:robustSimpleThres}]
    Firstly, in the non-strategic setting, the searcher achieves $\sfrac{1}{2}$-approximation since $T$ lies within the $\sfrac{1}{2}$-approximation threshold spectrum  (\Cref{defn:1/2 threshold}). Now, consider the case when the players are being strategic. Given any threshold stopping policy with threshold $T$, let $(G_1^T, \ldots, G_\boxNum^T)$ be 
    the players' optimal information revealing strategies. 
    We know that each $G_i^T$ must be a binary distribution such that $\supp(G_i^T) =\{a_i^T, T\}$, because $\priorMean_i \le \priorMean(\Bar{\prior}) \le T$ by definition of expected max. Let $t_i(T)$ be the pooling cutoff. $t_i(T)$ satisfies \[
    \int_{t_i(T)}^1 s \mathrm{d} {\prior_i(s)} = (1-\prior_i(t_i(T)))\cdot T
    \]
    Let $G(T)$ denote the probability that none of player $1$ through $\boxNum$ realizes a reward greater than $T$ in the strategic setting. We know 
    \[
    G(T) = \prod_{i=1}^\boxNum G_i^T(T) = \prod_{i=1}^\boxNum \prior_i(t_i(T)) \le \prod_{i=1}^\boxNum \prior_i(T) = \Bar{\prior}(T)
    \]
    where the inequality is because $\prior_i$ is non-decreasing and $t_i(T) \le T$. It follows that the payoff of the searcher in the strategic setting is \[
    \searcherUS(T) \ge (1-G(T))\cdot T \ge (1-\Bar{H}(T))\cdot T \ge \alpha \lambda(\Bar{\prior})
    \]
    Thus, this simple threshold policy is also an $\alpha$-approximation in the strategic setting. It follows that the policy is a $\alpha$-robust policy. 
\end{proof}

\begin{proof}[Proof of \Cref{lem:conditionForConvex}]
    Let each player $i$ have cdf $\prior_i$ and pdf $\priorPDF_i$, and denote $F_i = \log(\prior_i)$, which is well-defined since the pdf, and therefore the cdf, are log-concave. Recall that
    \begin{align*}
        \Bar{\prior}(x) = \prod_{i} \prior_i(x)
    \end{align*}
    Taking the log of both sides,
    \begin{align*}
        F(x) = \sum_{i} F_i(x)
    \end{align*}
    Taking the second derivative, we get
    \begin{align*}
        \frac{\Bar{\prior}''(x)\Bar{\prior}(x) - [\Bar{\prior}'(x)]^2}{[\Bar{\prior}(x)]^2} &= \sum_i F_i''(x)\\
    \end{align*}
    Rearranging terms,
    \begin{align*}
        \frac{\Bar{\prior}''(x)}{\Bar{\prior}(x)} &= \frac{\Bar{\prior}'(x)^2}{\Bar{\prior}(x)^2} + \sum_i F_i''(x) =(F'(x))^2 + \sum_i F_i''(x) =  (\sum_i F_i'(x))^2 + \sum_i F_i''(x)
    \end{align*}
    Note if we can show the right hand side is $\ge 0$, then $\Bar{\prior}''(x) \ge 0$, and thus $\Bar{\prior}$ is convex. To accomplish this, we just need to show for all $i$,
    \begin{align*}
        F_i''(x) + \sum_{j} F_i'(x)F_j'(x) &\ge 0 \\
        \Longleftrightarrow \frac{\priorPDF_i'(x)\prior_i(x) - \priorPDF_i^2(x)}{\prior_i^2(x)} + \sum_j \frac{\priorPDF_j(x)\priorPDF_i(x)}{\prior_j(x)\prior_i(x)} &\ge 0 \\ 
        \Longleftrightarrow \frac{\priorPDF_i'(x)}{\prior_i(x)} + \sum_{j \neq i} \frac{\priorPDF_j(x)\priorPDF_i(x)}{\prior_j(x)\prior_i(x)} &\ge 0 \\ 
        \Longleftrightarrow \frac{\priorPDF_i'(x)}{\priorPDF_i(x)} + \sum_{j \neq i} \frac{\priorPDF_j(x)}{\prior_j(x)} & \ge 0
    \end{align*}

    The last step holds because $h_i$ and $H_i$ are log-concave and thus positive. Since our distributions are from $F_{\alpha,\beta}$, $\priorPDF_i(1) \ge \alpha$ and $\priorPDF_i'(1) \ge -\beta$. Since $\priorPDF_i$ is log-concave, we know $\frac{\priorPDF_i'(x)}{\priorPDF_i(x)} = \frac{d}{dx} \log(\priorPDF_i(x))$ is non-increasing. Also, since $\prior_j$ is log-concave, we know $\frac{\priorPDF_j(x)}{\prior_j(x)} = \frac{d}{dx} \log(\prior_j(x))$ is also non-increasing. Hence for all $x$, 
    \begin{align*}
        \frac{\priorPDF_i'(x)}{\priorPDF_i(x)} + \sum_{j \neq i} \frac{\priorPDF_j(x)}{\prior_j(x)} \ge \frac{\priorPDF_i'(1)}{\priorPDF_i(1)} + \sum_{j \neq i} \frac{\priorPDF_j(1)}{\prior_j(1)} \ge \frac{\priorPDF_i'(1)}{\priorPDF_i(1)} + \sum_{j \neq i} \priorPDF_j(1) \ge \frac{\beta}{\alpha} + (\boxNum-1) \alpha \ge 0
    \end{align*}

    Therefore, we see that $\Bar{\prior}$ is convex, as desired.
\end{proof}

\end{document}